\documentclass[runningheads]{llncs}
\usepackage[utf8]{inputenc} 
\usepackage[T1]{fontenc}
\usepackage{hyperref}       
\usepackage{url}            
\usepackage{booktabs}       
\usepackage{amsfonts}       
\usepackage{nicefrac}       
\usepackage{microtype}      
\usepackage{xcolor}         
\usepackage{multirow}
\usepackage[algo2e,ruled,linesnumbered]{algorithm2e}

\usepackage{amsmath}
\usepackage{amssymb}
\usepackage{mathtools}
\usepackage{wrapfig}
\usepackage{graphicx}
\usepackage{thmtools} 
\usepackage{thm-restate}

\DeclareMathOperator*{\argmax}{arg\,max}
\DeclareMathOperator*{\argmin}{arg\,min}

\usepackage{caption}
\usepackage{subcaption}

\newcommand{\calX}{{{\mathcal{X}}}}
\newcommand{\calL}{{{\mathcal{L}}}}

\newcommand{\calWX}{{{\mathcal{WX}}}}

\newcommand{\calW}{{{\mathcal{W}}}}
\newcommand{\calR}{{{\mathcal{R}}}}

\newcommand{\cost}{{\mathit{cost}}}

\usepackage[ruled]{algorithm2e} 

\SetAlFnt{\small}
\SetAlCapFnt{\small}
\SetAlCapNameFnt{\small}
\SetAlCapHSkip{0pt}
\begin{document}
\title{An Adaptive and Verifiably Proportional Method for Participatory Budgeting}
\titlerunning{An Adaptive Method of Equal Shares}
%
\author{Sonja Kraiczy\and
Edith Elkind}
%
%
\institute{Department of Computer Science, University of Oxford,\\ 7 Parks Rd, Oxford OX1 3QG, United Kingdom 
\email{\{Sonja.Kraiczy,Edith.Ekind\}@cs.ox.ac.uk}\\
}
\maketitle              
\begin{abstract}
Participatory Budgeting (PB) is a form of participatory democracy in which citizens select a set of projects to be implemented, subject to a budget constraint. The Method of Equal Shares (MES), introduced in \cite{MES2}, is a simple iterative method for this task, 
which runs in polynomial time and satisfies a demanding proportionality axiom (Extended Justified Representation) in the setting of approval utilities. However, a downside of MES is that it is non-exhaustive: given an MES outcome, it may be possible to expand it by adding new projects without violating the budget constraint. To complete the outcome, the approach currently used in 
practice\footnote{e.g., in Wieliczka in Apr 2023, {\tt https://equalshares.net/resources/zielony-milion/}}
is as follows: given an instance with budget $b$, one searches for a budget $b'\ge b$ such that when MES is executed with budget $b'$, it produces a maximal feasible solution for $b$. The search is greedy, i.e., one has to execute MES from scratch for each value of $b'$. To avoid redundant computation, we introduce a variant of MES, which we call Adaptive Method of Equal Shares (AMES). Our method is budget-adaptive, in the sense that, given an outcome $W$ for a budget $b$ and a new budget $b'>b$, it can 
compute the outcome $W'$ for budget $b'$ by leveraging similarities between $W$ and $W'$. This eliminates the need to recompute solutions from scratch
when increasing virtual budgets. Furthermore, AMES satisfies EJR
in a certifiable way: given the output of our method, one can check in time 
$O(n\log n+mn)$ that it provides EJR (here, $n$ is the number of voters and $m$ is the number of projects).
We evaluate the potential of AMES on real-world PB data, showing that small increases
in budget typically require only minor modifications of the outcome.

\end{abstract}

\section{Introduction}
Participatory Budgeting (PB) gives residents of a city the power to decide how (part of a) public budget will be spent. Starting with Porto Alegre in Brazil, PB is now used in many cities around the world including locations in France, Iceland, Italy, Poland, Spain and many more \cite{wampler2021participatory}. Typically, the city council collects proposals for projects to be funded \cite{cabannes2004participatory}, such as building new cycle lanes or refurbishing a playground, and then
the residents vote to indicate which of the proposed projects they would like to see implemented. These votes are then aggregated into a final outcome, which must respect the budget constraint~\cite{aziz2021participatory}.
In practice, it is common to use the \textit{greedy rule}, which asks the residents which projects they approve, sorts the projects by the number of approvals (from highest to lowest), and adds projects one by one in this order until the budget is exhausted.
However, the greedy rule is very far from being proportional: if 51\% of residents approve
one set of projects, while 49\%  approve a disjoint set of projects, it will consider all
projects approved by the 51\% majority before those approved by the 49\% minority.

In contrast, the recently introduced Method of Equal Shares (MES)~\cite{MES} gives all voters equal voting power. The outcome of MES is proportional, as formalized by the demanding EJR (Extended Justified Representation) axiom \cite{aziz2017justified}. (This axiom was originally formulated for the simpler setting where all projects have equal cost, known as multiwinner voting, and then extended to the PB setting \cite{MES2}.) Indeed, MES is the first voting rule to satisfy the EJR axiom in the context of participatory budgeting. MES has already been used in practice in  Wieliczka (Poland) in April 2023, and in Aarau (Switzerland) in June 2023~\cite{mesweb}. Informally, it proceeds by (virtually) giving each voter an equal part of the budget and selecting projects iteratively one by one; the cost of each selected project is shared (almost) evenly among its supporters. In each iteration, MES selects a project so as to minimize the per-head price paid by its supporters:
for instance, it will prioritize a project that costs \$10,000 and is supported by 2000 voters over a project that costs \$8,000 and and is supported by 1000 voters.
It stops when none of the remaining projects can be afforded by its supporters.

\paragraph{Efficient Completion}{ 
A downside of MES is that it may fail to exhaust the budget, i.e., there may be projects that receive approvals from some voters and would fit within the remaining budget, but remain unselected. Therefore, a natural question is how to best complete the outcome of~MES. 
In practice, the current approach to the completion problem (for a detailed discussion, see \cite{mesweb}) is to run the method multiple times, increasing the (virtual) budget in each iteration. In more detail, suppose the real budget available is three million dollars, but the output of MES with such a budget is not exhaustive. Then we increase every voter's budget by one dollar and run MES with the increased budget. We proceed in this way until either the outcome is exhaustive or the next budget increase would
result in an outcome whose total cost is more than three million dollars. 

However, running the Method of Equal Shares from scratch repeatedly may be very inefficient: Intuitively, the outcome is unlikely to change much if the budget is increased by a small amount. Hence, it would be desirable to leverage the overlap between the outcomes of successive iterations, and construct the outcome iteratively, only modifying it to the extent necessary, instead of executing MES from scratch each time we increase the budget.
}

\paragraph{Efficient Verifiablility} 
{Another key challenge in multiwinner voting and participatory budgeting is efficiently verifying the proportionality of an outcome. In particular, verifiability is important in the context of blockchains, and specifically their applications in Decentralized Autonomous Organizations (DAOs) \cite{beck2018governance,sims2019blockchain}. DAOs are based on creating a large network of nodes belonging to human stakeholders, with each node running a protocol locally.
The vision is for such a network to offer services and make democratic decisions about various matters in a decentralized fashion, without the need to trust or rely on a central authority. The protocol should work even if the computational resources of the humans behind nodes are limited, such as off-the-shelf computers or smartphones, so that as many as possible can benefit from the network.

Some blockchain networks use multiwinner voting to appoint validators~\cite{burdges2020overview,brunjes2020reward,grigg2017eos}. Validators are special roles that nodes can take on: they have to validate transactions and receive a monetary reward for doing so (or get punished for adversarial behavior). 
It is desirable to select validators (from the set of candidate nodes) in a proportional manner, both to increase voting nodes' satisfaction \cite{dposgov}, and to avoid the centralization of power \cite{eoscentr}. Unfortunately, many proportional multiwinner rules are computationally hard \cite{PAV-hard} or else have prohibitively slow polynomial running time, so recent work by \cite{cevallos2021verifiably} proposes efficient verifiability as a solution. The key observation is that the (expensive) computational task
of choosing the validators can be performed by a non-trusted party (``off-chain'')
as long as the proportionality of the proposed solution can be efficiently checked by any node. Unfortunately, 
checking whether an arbitrary outcome satisfies the EJR axiom (or even the weaker PJR axiom) is NP-hard even in the setting of multiwinner voting~\cite{aziz2017justified,aziz2018complexity}. Nevertheless,  
\cite{cevallos2021verifiably} present a new multiwinner voting rule phragmms, such that 
the outputs of this rule satisfy PJR and this can be efficiently verified in time linear in the input size
(the rule outputs auxiliary information, which can be used as a certificate of PJR).
Now, verifiability remains a relevant concern in the broader context
of participatory budgeting: it is natural for DAOs to allocate funds
for new projects based on stakeholder votes. However, 
prior to our work it was not known whether there exist voting
rules for participatory budgeting that satisfy demanding proportionality axioms
(such as EJR) in a verifiable way.
}

\smallskip

\paragraph{Our Contribution}{
We present a method for participatory budgeting with approval utilities that is closely related to the Method of Equal Shares. Specifically, we define a weak order $\rhd$ on feasible solutions (i.e., pairs of the form $(W, X)$, where $W$ is the selected set of projects and $X$ describes how the costs of these projects are shared among the voters) that is inspired by the definition of the MES rule. Our method, which we call the Adaptive Method of Equal Shares (AMES), operates by executing greedy local search with respect to $\rhd$, i.e., it outputs solutions that
cannot be improved with respect to $\rhd$ by simple transformations. 
Interestingly, in all solutions output by AMES the cost of each project
is shared {\em exactly} equally among its supporters. Moreover, if a MES solution
for a given instance has this `exact equal sharing' property, it is also in the 
output of AMES on this instance, i.e., the two methods are indeed similar.

We show that our method has several desirable properties:
\begin{enumerate}
\item AMES satisfies the EJR axiom for approval utilities.
\item AMES runs in time $O(mn\log n+mn)$, 
where $m$ is the number of projects and $n$ is the number of voters.
Moreover, it admits a very efficient algorithm for verifying that its output satisfies EJR: by using $O(mn)$ auxiliary data, it can verify EJR in time $O(mn)$, i.e., linear in the input size.
\item AMES is budget-adaptive: To determine its output for budgets $b_1<b_2\ldots<b_t$, instead of having to start from scratch for every single budget, we can use the output for budget $b_i$ to obtain the output of the rule for budget $b_{i+1}$.
\end{enumerate}
This combination of properties in unique: AMES is the first polynomial-time voting method to satisfy EJR in the approval PB setting that is budget-adaptive and verifiably proportional. Indeed, in the PB setting Phragm\'en's rule \cite{phragmen:p1,janson2016phragmen} is naturally budget-adaptive, but fails EJR even for multiwinner voting (i.e., for PB with unit costs), while Local Search PAV \cite{aziz2018complexity} is fully adaptive, but is known to fail EJR in the PB setting~\cite{MES2}. For multiwinner setting, Brill and Peters~\cite{brill2023robust}, have recently proposed a new axiom, which they call EJR+. This axiom
strengthens EJR, can be verified in time $O(kmn)$ in general (where $k$ is the size of the committee), and is satisfied by a simple greedy rule.
However, the analysis in~\cite{brill2023robust} does not show linear-time verifiability; moreover, the verifiability results of \cite{brill2023robust}, just as those of \cite{cevallos2021verifiably}, only apply to multiwinner approval voting rather than the more general setting of participatory budgeting.

To gain insights on the practicality of our method, we analyze real-world participatory budgeting data from several cities in Poland, and consider the average number of changes in the outcome as a result of increases in the budget. Our findings confirm that AMES, and the budget-adaptive approach in general, has great potential to make practical implementation much faster.

We also show that, by handling ties carefully, we can ensure that the output 
of the algorithm run from scratch with budget $b$ is identical to the output
of the adaptive algorithm that starts from a feasible solution for a budget $b'<b$ and then (gradually) raises the budget to $b$. This property is attractive, because
it ensures that AMES is as easy to explain to the voters as the standard MES rule: instead of explaining how the algorithm modifies the solution with each budget increase, one can simply demonstrate its execution for the final budget $b$.
\smallskip

\noindent\textbf{Related Work.}
The special case of participatory budgeting with approval ballots 
where all projects have the same cost is known {\em as approval-based multiwinner voting}, and there is a very substantial literature on approval-based multiwinner voting rules, their axiomatic properties, and algorithmic complexity \cite{LS23}.
The EJR axiom was first introduced in this context in \cite{aziz2017justified}, and the first polynomial-time rule shown to satisfy it was Local Search PAV \cite{aziz2018complexity}, followed by the arguably more natural Method of Equal Shares \cite{MES}. 

In recent years, there has been an explosion of interest in participatory budgeting; by now there is a wide variety of different models and approaches \cite{aziz2021participatory}. 
Proportionality axioms for participatory budgeting were first considered by \cite{aziz2018proportionally}. In \cite{MES2}.
the authors extend the Method of Equal Shares from the multiwinner voting setting \cite{MES} to a general PB model with general additive utilities. They show that MES satisfies EJR up to one project, giving the first such polynomial-time voting method for participatory budgeting. For the special case of participatory budgeting with approval utilities (which is the focus of our work), \cite{MES2} show that MES satisfies EJR.

Our adaptive method uses a stability notion very similar to that of Peters et al.~\cite{peters2021market}, in that they also consider justifying outcomes of an election by means of a price system that satisfies the stability condition.
Peters et al. are concerned with the existence and properties of this stability notion for a given committee size, which may not always exist, while our formulation is not tied to a committee size. Instead, we show that we can use the stability notion to efficiently move through the space of intermediate outcomes of MES when increasing the committee size/ budget.



\section{Preliminaries}
For each $t \in \mathbb{N}$, we write $[t] = \{1,2,\ldots, t\}$. 
\paragraph{Participatory Budgeting}{ We first introduce the model of participatory budgeting with approval ballots. 
An {\em election} is a tuple $E = (N,P,(A(i))_{i\in N}, b, \cost)$, where:
\begin{enumerate}
\item $N = [n]$ and $P = \{p_1, \ldots, p_m\}$ are the sets of {\em voters} and 
{\em projects}, respectively;
\item for each $i\in N$ the set $A(i)\subseteq P$ is the {\em ballot} of voter $i$, 
i.e., the set of projects approved by~$i$;
\item $b \in \mathbb{Q}_{> 0}$ is the available budget;
\item $\cost:P \rightarrow \mathbb{Q}_{> 0}$ is a function that for each $p \in P$ indicates the cost of selecting $p$. For each $W \subset P$, we denote the total cost of $W$ by $\cost(W) = \sum_{p \in W} \cost(p)$.
\end{enumerate}
We assume every project is approved by at least one voter.

An {\em outcome} is a set of projects $W\subseteq P$ that is {\em feasible}, i.e., satisfies  $\cost(W) \leq b$. 
We denote the set of all outcomes for an election 
$E = (N,P,(A(i))_{i\in N}, b, \cost)$ by $\calW(E)$; sometimes,  
to emphasize the dependence on the budget $b$, we write $\calW(b)$ instead of $\calW(E)$.
We say that an outcome $W$ is {\em exhaustive} if it is a maximal feasible set of projects, i.e., $W\cup\{p\}\not\in\calW(E)$ for each $p\in P\setminus W$. We assume that the voters have approval utilities, i.e., the utility of voter $i\in N$
from an outcome $W\subseteq P$ is given by $|A(i)\cap W|$. Our goal is to select
an outcome based on voters' ballots.
An {\em aggregation rule} (or, in short, a {\em rule}) is a function $\calR$ that for each election $E$ selects an outcome $\calR(E)\in\calW(E)$,
called the {\em winning outcome}.}
\smallskip
\paragraph{Load Distribution}{Given an election $E = (N,P,(A(i))_{i\in N}, b, \cost)$,
a {\em load distribution} for an outcome $W\in\calW(E)$ is a collection of rational numbers 
$X = (x_{i, p})_{i\in N, p\in P}$ that form
a feasible solution to the following linear program:
\begin{align}
    &\sum_{i\in N}x_{i, p} = \cost(p)&\text{for all $p\in W$}\\
    &0\le x_{i, p}\le \cost(p) &\text{for all $i\in N, p\in W$}\\
    &x_{i, p} =0 &\text{for all $i\in N$, $p\in P\setminus (W\cap A(i))$}
\end{align}
Given a load distribution $X$ and a voter $i\in N$, we write $X_i=\sum_{p\in P}x_{i, p}$
to denote the total load of~$i$, and 
write $N_p(X)$ to denote the set of voters who pay for $p$ in $X$: 
$$
N_p(X)=\{i\in N: x_{i, p}>0\}.
$$
We denote the set of all load distributions for an outcome $W$ by $\calX(W)$. 
Further, given an election $E = (N,P,(A(i))_{i\in N}, b, \cost)$, we set 
$$
\calWX(E)=\{(W, X): W\in\calW(E), X\in\calX(W)\};
$$
just as before, we write $\calWX(b)$ instead of $\calWX(E)$ to emphasize
the dependence on $b$, and
refer to elements of $\calWX(b)$ as {\em solutions for $b$}
(or simply {\em solutions}).
Intuitively, the linear program in Equations (1)--(3) describes how voters in $N$ can share the cost of projects in $W$; each project is associated with a (rational) cost (or, load), and the cost of a project can only be shared by voters who approve that project. Note that this linear program is feasible as 
long as we assume that every project in $W$ has at least one approval.

We will say that a load distribution $X\in\calX(W)$ is {\em priceable} \cite{MES2} if $\sum_{p\in W} x_{i,p}\leq \frac{b}{n}$ for all $i\in N$. We say that $X$ is {\em equal-shares} if it is priceable and for every $p\in W$ and every pair of voters $i, j\in N_p(X)$ we have $x_{i,p}=x_{j,p}$, that is, the voters who pay for $p$ share the cost of $p$ exactly equally. We will say that a solution $(W,X)$ is equal-shares (resp., priceable) if $X$ is an equal-shares (resp., priceable) load distribution in $\calX(W)$; an outcome $W$ is equal-shares (resp., priceable) if there exists an $X\in\calX(W)$ such that $(W, X)$
is equal-shares (resp., priceable). We denote the set of all equal-share solutions 
for an election $E$ by $\calWX^=(E)$; again, we will sometimes write $\calWX^=(b)$ 
instead of $\calWX^=(E)$ to emphasize the dependence on $b$.}

\begin{example}\label{ex:run}
We will use the following
instance of participatory budgeting as a running example:
\begin{align*}
E &=(N, P, (A(i))_{i\in N},b, \cost) \text{ where }\\
N &=\{1, 2, 3\}, P = \{p_1, p_2, p_3, p_4, p_5\}, b=35, 
A(1) =A(2)=A(3)= P,\\
\cost(p_j) &=6\text{ for }j\in \{1, 2, 3\}, \cost(p_4)=7, \cost(p_5)=10.
\end{align*}
We focus on the outcome $W=P$; note that $\cost(W)=35=b$.

Consider the load distribution $X$ given by
$x_{i, p}=\cost(p)/3$ for all $i\in N$, $p\in P$.
Note that $X$ is priceable and equal-shares
and hence $(W, X)\in\calWX^=(E)$.

In contrast, consider the load distribution $X'$
given by $x'_{1, p}=35/9$ for $p\in \{p_1, p_2, p_3\}$,
$x'_{2, p_4}=7$, $x'_{2, p_1}=x'_{2, p_2}=19/9$, $x'_{2, p_3}=4/9$,
$x'_{3, p_5}=10$, $x'_{3, p_3}= 5/3$ (all values not explicitly specified are $0$).
This load distribution is priceable: we have $X_i=35/3$ for each $i\in N$.
However, it is not equal-shares: all three voters make different contributions towards $p_3$. 

Moreover, the load distribution $X''$ given by $x''_{1, p}=6$
for $p\in\{p_1, p_2, p_3\}$, $x''_{2, p_4}=7$, $x''_{3, p_3}=10$
(all values not specified are $0$) is not priceable: we have $X_1=18 >35/3$.
Nevertheless, it is a load distribution for $W$: all projects in $W$ 
are fully paid for, and each voter only pays for projects in $W$ she approves. 
\end{example}

\smallskip
\paragraph{Extended Justified Representation (EJR)}{
Given a participatory budgeting election $E = (N, P, (A(i))_{i\in N}, b, \cost)$ 
and a subset of projects $T\subseteq P$, 
we say that a group of voters $S$ is {\em $T$-cohesive} if 
$\frac{|S|}{n} \geq \frac{\cost(T)}{b}$ and $T \subseteq \cap_{i\in S} A(i)$. 
An outcome $W\in\calW(E)$ is said to provide {\em Extended Justified Representation (EJR)} for approval utilities if for each $T \subset P$ and each $T$-cohesive group $S$ of voters there exists a voter $i \in S$ such that $|A(i) \cap W| \geq |T|$.
%
%
A rule $\calR$ satisfies {\em Extended Justified Representation} 
if for each election $E$ the outcome $\calR(E)$ provides EJR. 

}
\smallskip
\paragraph{Method of Equal Shares}{ The Method of Equal Shares (MES) is defined for the general PB model, in which voters may have real-valued additive utilities~\cite{MES,MES2}. In this work 
we focus on approval utilities only, and hence we describe MES for approval utilities.

Fix an election $E = (N, P, (A(i))_{i\in N}, b, \cost)$. Initially, 
each voter $i\in N$ is allocated a fixed budget of $\frac{b}{n}$.
MES builds a solution $(W, X)$ for $b$. 
It starts by setting $W=\varnothing$, $x_{i, p}=0$ for all $i\in N$, $p\in P$.
At each iteration, MES selects one project to be added, i.e., it sets 
$W\gets W\cup\{p\}$ for some project $p$ and updates the load distribution $X$ 
so as to pay for $p$. The cost of $p$ is shared equally among its supporters, with the following exception: if, by contributing her share of the cost, the voter would have to exceed her budget $\frac{b}{n}$, she simply contributes her entire remaining budget instead. Formally, 
given the current outcome $W$ and load distribution $X$, 
for each project $p\in P\setminus W$ we compute the quantity 
$$
\rho(p) = \min \left\{\rho: \sum_{i:p\in A(i)}
                \min\left\{\rho, \frac{b}{n}-X_i\right\} = \cost(p)\right\},
$$
add a project $p$ with the smallest value of $\rho(p)$ to $W$, 
and set $x_{i, p} = \min\{\rho(p), \frac{b}{n}-X_i\}$
if $p\in A(i)$ and $x_{i, p} = 0$ otherwise. If $\rho(p) = +\infty$
for all projects in $P\setminus W$, the algorithm terminates. 
Importantly, this may happen even if $W$ is not exhaustive.

\section{Stable Equal-Shares Solutions Allow for Fast Verification of EJR}\label{sec:mes}
We will now define a notion of stability for equal-shares solutions; our definition is designed to ensure that every stable outcome satisfies EJR.

The intuition behind our definition is as follows. If a solution $(W, X)$ satisfies $X_i<\frac{b}{n}$, voter $i$ is willing to contribute her remaining budget to support further projects in $A(i)$. Moreover, even if $i$ does not have enough budget left to
contribute to new projects, she may still prefer to spend money on more cost-efficient projects. To achieve this, voter $i$ may want to withdraw her support from a project and 
reallocate it to a more cost-efficient project.

To formalize this intuition, we first define a weak order 
on the set of equal-shares solutions for a given election.

\paragraph{Per-voter price}
{Given an election $E = (N,P,(A(i))_{i\in N},b,\cost)$, an outcome $W\in\calW(E)$
and an equal-shares load distribution $X\in\calX(W)$, 
we define
the {\em per-voter price} of a project $p\in W$ with respect to $X$ as the 
cost of $p$ divided by the number of voters who pay for $p$ in $X$:
$$
\pi(p, X) = \frac{\cost(p)}{|N_p(X)|};
$$
note that, since $(W, X)$ is an equal-shares solution, this is exactly
the amount that each of the voters who pays for $p$ contributes towards 
the cost of $p$.
If $p\notin W$, we set $\pi(p, X)=+\infty$.
The {\em per-voter price vector} of a load distribution $X\in\calX(W)$ is the list of numbers $\boldsymbol{\pi}(X) = (\pi(p, X))_{p\in P}$, sorted from the smallest to the largest, with ties broken based on a fixed order of projects in $P$. We write $\pi_j(X)$ to denote the $j$-th entry of $\boldsymbol{\pi}(X)$.
Given two solutions $(W, X), (W', X')\in \calWX^=(E)$,
we write $(W, X)\rhd (W', X')$ if 
$\boldsymbol{\pi}(X)\preceq_\text{lex} \boldsymbol{\pi}(X')$, 
where $\preceq_\text{lex}$ is the lexicographic order on $m$-tuples.
Note that for every pair of solutions $(W, X), (W', X')\in\calWX^=(E)$
we have $(W, X)\rhd (W', X')$ or $(W', X')\rhd (W, X)$, i.e., $\rhd$
is a weak order on $\calWX^=(E)$.
}

\paragraph{Stability }
We are now ready to define our notion of stability.
Consider an election $E=(N, P, (A(i))_{i\in N}, b, \cost)$ and a 
solution $(W, X)\in \calWX^=(E)$. For each voter $i\in N$, 
let $z_i=\max\{x_{i, p}: p\in W\}$ 
so $z_i\ge 0$ if $i\in N_p(X)$ for some $p\in W$ and $z_i=-\infty$ otherwise. 
Let $K=\{\frac{\cost(p)}{i}: p\in P, i\in N\}$ and let $\epsilon=\min_{x,y\in K, x\neq y}|x-y|$ 
denote a lower bound on the minimum possible positive difference between two payments in $X$. Let
\begin{equation}\label{eq:kappa}
\kappa_i(X)=\max\left\{z_i-\epsilon, \frac{b}{n}-X_i\right\};
\end{equation}
we refer to the quantity $\kappa_i(X)$ as the {\em capacity} of voter $i$ (and omit $X$ from the notation when it is clear from the context).
Intuitively, $\kappa_i(X)$ is the amount that $i$ is prepared to contribute 
towards a new project that she approves, 
either by using up her unspent budget
(as captured by the $\frac{b}{n}-X_i$ term), or by withdrawing her contribution from   one of the projects with the highest per-voter price among the ones 
she is currently paying for, and using 
these funds to pay for a less expensive project (i.e., one with per-voter price at most $z_i-\epsilon$). Note that this includes the scenario in which voters decrease their contribution to a project by sharing its load with an increased number of voters.
\begin{definition}\label{def:stable}
 A solution $(W, X)\in\calWX^=(E)$ is {\em unstable} if there is a project 
 $p\in P$ and a positive integer $t$ with $t > |N_p(X)|$
 such that $|\{i\in N: p\in A(i), \kappa_i(X) \ge \frac{\cost(p)}{t}\}|\ge t$;
 otherwise we say that
$(W, X)$ is {\em stable}. 
\end{definition}
We denote the set of all stable equal-shares solutions for an election $E$ by $\calWX^*(E)$, 
and write $\calWX^*(b)$ instead of $\calWX^*(E)$ when we want to emphasize the dependence on the budget $b$.


\subsection{Properties of Stable Outcomes: Proportionality and Verifiability}\label{sec:stab+ejr}

Our primary goal in this section is to show that stable outcomes provide EJR. To this end, we establish a property that may be of independent interest (we defer the proof, 
as well as some of the subsequent proofs, to the full version of the paper).

\begin{restatable}{lemma}{lemrestrict}\label{lem:restrict}
Given an election $E$, a subset of voters $V\subseteq N$ and a subset of projects $T\subseteq P$, consider the election $E'=(V, T, (A'(i))_{i\in V}, b, \cost')$ with 
$A'(i)=A(i)\cap T$ for all $i\in V$ and $\cost'(p)=\cost(p)$ for all $p\in T$.
Fix a pair of solutions 
$(W, X)\in\calWX^*(E)$, $(W',X')\in\calWX^*(E')$ and
a voter $i\in V$, and let $r=|\{p: x'_{i, p}>0\}|$.
Then there exists a voter $j\in V$ such that $|(A(i)\cup A(j))\cap W|\geq r$. In particular, there exists a set of projects $P'$ with $P'\subseteq A(i)\cap W'$ such that $P'\subseteq W$ 
and voter $j$ approves at least $r-|P'|$ projects in $W\setminus P'$.
\end{restatable}

\noindent We use Lemma~\ref{lem:restrict} to show that 
if $(W,X)\in\calWX^*(E)$ then $W$ provides EJR.
 \begin{theorem}\label{thm:ejr}
If $(W,X)$ is a stable solution for an election $E$, then $W$ provides EJR for approval utilities.
\end{theorem}
\begin{proof}
Consider a set of projects $T\subseteq P$ and a group of voters $V$
of size at least $\frac{\cost(T)}{b}\cdot n$ 
such that $T\subset A(i)$ for all $i\in V$.
We will show that $|A(i)\cap W| \geq |T|$ for some $i\in V$.
Consider the restricted instance $E'$ where all approvals but those from $V$ to $T$ 
are removed, and
the equal-shares solution $(T,X')$, where $x'_{i,p}=\frac{\cost(p)}{|V|}$ 
for every $i\in V$ and $p\in T$, and all other entries of $X'$ are zero.
Note that every voter $i\in V$ approves $|T|$ projects in $T$.
Further, $(T,X')$ is in $\calWX^*(E')$, 
since all projects that are approved by some voter in $E'$ are selected, 
and the cost of each project is shared by all $|V|$ voters.
Consider an arbitrary voter $i\in V$. 
By Lemma~\ref{lem:restrict}, since $(W, X)$ is a stable solution for the original election $E$, there is a $T'\subset A(i)\cap T$ with $T'\subseteq W$ such that some voter $j\in V$ approves additional $|T|-|T'|$ projects in $W\setminus T'$. But since $j$ approves all projects in $T'\subset T$ as well, it follows that she approves
at least $|T|$ projects in $W$, as desired.
\qed\end{proof}
Next, we will show that stability can be verified very efficiently. Since stability implies EJR, it follows that
for any PB rule that outputs stable solutions (so in particular for the AMES rule, to be defined in the next section), we can quickly verify 
that the associated outcome provides EJR. Importantly, this verification procedure may be much faster than evaluating the rule on a given election. 

This result may appear
counterintuitive, since  checking whether a given outcome provides EJR is 
known to be NP-hard, even in the context of multiwinner voting~\cite{aziz2017justified}.
This apparent contradiction is resolved by observing that the input to our verification
procedure is a solution, i.e., a pair $(W, X)$, rather than an outcome $W$:
the auxiliary information provided by $X$ enables the verification algorithm
to run in nearly-linear time.

\begin{restatable}{theorem}{propverifycomplexity}\label{prop:verifycomplexity}
There is a verification process ${\mathcal V}(W,X)$ that, given an election
$E=(N, P, (A(i))_{i\in N}, b, \cost)$ with $|N|=n$, $|P|=m$ 
and a solution in $\calWX(E)$, 
decides the stability of $(W,X)$ and runs in time $O(n\log n +mn)$.
\end{restatable}
 
In fact, the proof of Theorem~\ref{prop:verifycomplexity} shows that we can achieve
linear-time verifiability if we can pass the list of capacities $\mathcal{K}=(\kappa_i(X))_{i\in N}$ sorted in non-increasing order as auxiliary data to the verifier 
(in addition to the solution $(W,X)$). 
More precisely, there exists a verification algorithm ${\mathcal V}(W,X,\mathcal{K})$ that decides stability in time linear in the size of the election (which matches the runtime of verifying that an output of phragmms satisfies PJR \cite{cevallos2021verifiably}). This is because ${\mathcal V}(W,X,\mathcal{K})$ does not have to sort the capacities; instead, it checks that the capacities have been computed correctly and verifies that $\mathcal{K}$ is non-increasing, all in time $O(nm)$.

\begin{restatable}{theorem}{propverifycomplexity2}\label{prop:verifycomplexity2}
There is a verification process ${\mathcal V}(W,X,\mathcal{K})$ that, given an election
$E=(N, P, (A(i))_{i\in N}, b, \cost)$ with $|N|=n$, $|P|=m$, 
a solution in $\calWX(E)$, and capacities $\mathcal{K}=(\kappa_i(X))_{i\in N}$ sorted in non-increasing order
decides the stability of $(W,X,\mathcal{K})$ and runs in time $O(mn)$.
\end{restatable}

To conclude this section, we mention a subtlety in the definition of stability. In our definition, we deem an outcome $(W,X)$ unstable if there is a project $p\in W$ such that the number of voters paying for $p$ could be increased. In the full version of the paper, we consider a weaker notion of stability, which only considers projects not currently included in the outcome. We show that this alternative notion of stability does not in general imply EJR; however, it \textit{does} imply EJR in the special case of multiwinner voting (i.e., unit-cost projects). 


\section{An Adaptive Method of Equal Shares}\label{sec:AMES}

Our notion of stability suggests a natural adaptive algorithm, which, given an 
equal-shares solution $(X, W)$, performs greedy update steps until it reaches a stable outcome.


\paragraph{Greedy update step}
{ If $(W, X)$ is unstable, it admits an {\em update step} 
$(W, X)\xrightarrow{p} (W', X')$, defined as follows. We identify a project $p$ and a size-$t$ set of voters
$V = \{i\in N: p\in A(i), \kappa_i(X)\ge \frac{\cost(p)}{t}\}$ that witness
the instability of $(W, X)$. 
Voters in $V$ all approve $p$; if each of them contributes $\frac{\cost(p)}{|V|}$ then $p$ will be fully paid for. However, if $\kappa_i(X)=z_i-\epsilon$, increasing $i$'s load by $\frac{\cost(p)}{|V|}$ may result in their load exceeding $\frac{b}{n}$. To overcome this, we determine a set of projects $W^-$ to be removed from~$W$. 

To this end, we initialize $W^-=\varnothing$,
and iterate through voters in $V$. For each voter $i\in V$, 
if $\frac{\cost(p)}{|V|} \le \frac{b}{n}-X_i$, we do nothing:
this voter can afford to pay for $p$ from her remaining budget.
Similarly, if $p\in W$ and $x_{i,p}>0$, we do nothing: after the update, 
this voter's contribution towards $p$ will be smaller than $x_{i, p}$.
If $\frac{\cost(p)}{|V|} > \frac{b}{n}-X_i$ and $x_{i,p}=0$, 
we let $q$ be a project with the highest per-voter price 
among the ones that $i$ contributes to, i.e., 
and $\pi(q, X)= \max_{q'\in P} x_{i, q'}$,
and set $W^-\gets W^-\cup\{q\}$ (of course, 
it may be the case that $q$ is already in $W^-$). After we process
all voters in $V$, we set $W' =  (W\setminus W^-)\cup\{p\}$, 
and define $X'$ by setting 
\begin{align*}
&x'_{i, p}=\frac{\cost(p)}{|V|}&\text{ for all } i\in V,\\ 
&x'_{i, p}= 0 &\text{ for all } i\in N\setminus V,\\
&x'_{i, q}=0 &\text{ for all } i\in N\text{ and } q\in W^-,\\
&x'_{i, q}=x_{i, q} &\text{ for all } i\in N, q\in P\setminus(W^-\cup\{p\}). 
\end{align*}
Note that, by construction, if a solution is stable, it does not admit
an update step. Moreover, we say that 
an update step $(W, X)\xrightarrow{p} (W', X')$ is {\em greedy}
if for every other update step $(W, X)\xrightarrow{q} (\overline{W}, \overline{X})$
we have $\pi(p, X')\le \pi(q, \overline{X})$. 
That is, a greedy update step selects a project
(by adding it or increasing its number of contributors)
so as to minimize the per-voter price.
}
The following proposition summarizes the key properties of an update step.

\begin{restatable}{proposition}{propupdate}
    
\label{prop:update}
Suppose $(W, X)\in\calWX^=(E)$, and $(W', X')$
is obtained from $(W, X)$ by an update step. 
Then $(W', X')\in\calWX^=(E)$ and $(W', X')\rhd (W, X)$.
\end{restatable}
\noindent Our notion of an update step now suggests the following procedure.

Start with an arbitrary pair $(W, X)\in\calWX^=(E)$,
and execute a sequence of greedy update steps from 
$(W, X)$ until a stable solution $(W^*, X^*)$ is reached. A pseudocode description of this iterative  algorithm, which we call the Adaptive Method of Equal Shares (AMES), 
is given in Algorithm~\ref{alg:ames}. We say that we run AMES \textit{from scratch} if the starting solution $(W,X)$ satisfies $W=\varnothing$, $x_{i, p}=0$
for all $i\in N, p\in P$.

In more detail, in each iteration our algorithm loops over all projects, and checks
if there is a project for which the number of supporters can be increased (lines 5--14);
we use the convention that $\max\varnothing = -\infty$. 
It also keeps track of the best such project ($p^*$), 
as measured by the price per voter (lines 9--11).
If some such project has been identified, the algorithm 
iterates through all voters whose contribution towards $p$ 
is about to increase in a way that exceeds their remaining budget;
for each such voter, it identifies one project that this voter is currently 
paying for, and adds it to the set of projects to be removed from $W$ (lines 16--19). 
It then updates $W$ and $X$ accordingly (lines 20--23).

\begin{algorithm2e} \SetAlgoNoLine
\caption{AMES (Adaptive Method of Equal Shares)}\label{alg:ames}
  \KwIn{$E = (N,P,(A(i))_{i\in N}, b, \cost), (W,X) \in \calWX^=(E)$}
  \KwOut{$(W^*,X^*) \in \calWX^*(E)$}
  \Repeat{\ $\textrm{flag} = {\tt false}$}{
    $(\kappa_i)_{i\in N} \gets \text{capacities}(W, X)$ (as per~\eqref{eq:kappa})\\
    $\textit{flag} \gets {\tt false}$\\
    $\pi^* = +\infty$\\
    \For{$p\in P$}{
        $t_p \gets \max\{t\in\mathbb N: t> |N_p(X)|, 
            |\{i\in N: p\in A(i), \kappa_i \ge \cost(p)/t\}|
                     \ge t\}$\\
        \If{$t_p\neq-\infty$}
          {
          $\textit{flag} \gets {\tt true}$\\
          \If {$\cost(p)/t_p < \pi^*$}
              {$\pi^*\gets \cost(p)/t_p$\\ 
               $p^*\gets p$
              }
          }
    }
    \If{$\text{flag} = {\tt true}$}
     {
       $S \gets \{i\in N: p^*\in A(i), x_{i,p^*}=0 \text{ and } \pi^*>\frac{b}{n}-X_i\}$\\
        \For{$i\in S$}{
          Let $p^-(i)$ be some project from $\argmax_{p\in W}\{x_{i,p}\}$
        }
        $W^- \gets \{p^-(i): i\in S\}$\\
        $W\gets (W \cup \{p^*\})\setminus W^{-}$\\
        $x_{i,p^*}\gets \pi^*$ for all $i\in N$ such that $p\in A(i)$ and $\kappa_i\ge \pi^*$\\
        $x_{i,p}\gets 0$ for all $p\in W^-,i\in N$\\
     }
  }
  \Return{$(W, X)$}
\end{algorithm2e}
We will use the instance $E$ from Example~\ref{ex:run} to illustrate the different kinds of update steps AMES can perform.

Suppose the priority order over the projects is $p_1>p_2>p_3>p_4>p_5$.
Consider the stable outcome $(W,X)$ given by 
$W=\{p_1,p_2,p_3,p_4\}$,  
$x_{1,p}=x_{2,p}=3$ for $p\in\{p_1, p_2, p_3\}$, $x_{1,p}=x_{2,p}=0$
for $p\in \{p_4,p_5\}$, 
$x_{3,p_4}=7$ and  $x_{3,p}=0$ for $p\neq p_4$.
When initialized on $E$ and $(W,X)$, AMES performs the following steps: 

\begin{itemize}
\item[(1)] $(W,X)\xrightarrow{p_1} (W^1,X^1)$
where 
$$
W^1=W, x^1_{i,p_1}=2\text{ and }x^1_{i,p}=x_{i,p} \text{ for } p\neq p_1, i\in N.
$$
\item[(2)] $(W^1,X^1)\xrightarrow{p_2} (W^2,X^2)$.
where 
$$
W^2=W^1, x^2_{i,p_2}=2\text{ and }x^2_{i,p}=x^1_{i,p} \text{ for } p\neq p_2, i\in N.
$$
\item[(3)] $(W^2,X^2)\xrightarrow[p_4]{p_3} (W^3,X^3)$
where 
\begin{align*}
&W^3=W^2\setminus\{p_4\}, x^3_{i,p_3}=2\text{ for $i\in [3]$}, x^3_{3,p_4}=0,\\
&\text{ and } x^3_{i,p}=x^2_{i,p}\text{ for all other }(i, p)\in N\times P
\end{align*}
\item[(4)]  $(W^3,X^3)\xrightarrow{p_4} (W^4,X^4)$
where 
$$
W^4=W^3\cup\{p_4\}, x^4_{i,p_4}=\frac{7}{3}\text{ and }x^4_{i,p}=x^3_{i,p}\text{ for }p\neq p_4,i\in N.
$$
\item[(5)] $(W^4,X^4)\xrightarrow{p_5} (W^5,X^5)$ where 
$$
W^5=W^4\cup\{p_5\}, x^5_{i,p_5}=\frac{10}{3}\text{ and  }x^5_{i,p}=x^4_{i,p}\text{ for }p\neq p_5,i\in N.
$$

\end{itemize}
Update steps (1) and (2) illustrate increasing the number of supporters of a project that is already included in the solution.
Update step (3) additionally illustrates the removal of a project ($p_4$), which enables the algorithm to increase the number of supporters of $p_3$.
Finally, update steps (4) and (5) illustrate simple addition of new projects.
Note that after step (2) the transition $(W^2,X^2)\xrightarrow[p_3]{p_4} (W',X')$ would be another valid update step, where 
\begin{align*}
W'=W^2\setminus\{p_3\}, x'_{i,p_4}=\frac{7}{3} \text{ for $i\in N$ }, 
x'_{j,p_3}=0\text{ for }j\in[2], \text{ and }\\
x'_{i,p}=x^2_{i,p}\text{ for all other }(i, p)\in N\times P,
\end{align*}
but it is not greedy, since the update in step (3) is strictly better
according to $\rhd$.

\smallskip
The next lemma shows that the per-voter prices of projects are non-decreasing in the order in which they were added.
\begin{restatable}{lemma}{lemdecappeal}
    
\label{lem:dec-appeal}
Consider two consecutive greedy update steps 
$(W, X)\xrightarrow{p'} (W', X')$ and
$(W', X')\xrightarrow{p''} (W'', X'')$.
Suppose that $\pi(p',X')<\pi(p',X)$ and $\pi(p'',X'')<\pi(p'',X')$. Then 
$\pi(p', X')\le \pi(p'', X'')$.
\end{restatable}

Lemma~\ref{lem:dec-appeal} is the key to showing that, by performing greedy update steps,
AMES converges quickly. 
\begin{restatable}{proposition}{propgreedy}\label{prop:greedy-k}
Suppose that AMES is executed on input $(W, X)$ and outputs a solution $(W^*, X^*)$.
Then it executes at most $|\{p\in P: \pi(p,X^*)<\pi(p,X)\}|$ greedy update steps. 
\end{restatable}
\noindent 
The proof of Proposition~\ref{prop:greedy-k} shows that 
in each update step we either add a project from $W^*\setminus W$ 
or else lower the per-voter price of an existing project in $W$.
A consequence of this result is that
we can compute a stable solution in polynomial time, 
by starting from an arbitrary equal-shares 
solution and performing greedy update steps.

\begin{restatable}{theorem}{proplsoptmesruntime}\label{prop:lsoptmes-runtime}
Given an election $E = (N,P,(A(i))_{i\in N}, b, \cost)$ and an outcome $(W,X)\in\calWX^=(E)$,
Algorithm~\ref{alg:ames} outputs $(W^*,X^*)\in \calWX^*(E)$ 
by performing at most $|\{p\in P: \pi(p,X^*)<\pi(p,X)\}|=O(m)$ 
update steps; each update step can be completed in time $O(n\log n+ nm)$.
\end{restatable}


\subsection{On Using AMES Adaptively}\label{sec:ames-consistency}
An advantage of our notion of stability (Definition~\ref{def:stable}) is that it 
is easily checkable, and hence one can quickly verify that
an outcome of AMES provides EJR. In particular, this means that 
the execution of AMES can be outsourced to a non-trusted, but computationally
powerful third party. On the other hand, it is arguably easier to explain that a solution was obtained by running AMES from scratch with a virtual budget $b'$, 
as opposed to being obtained by computing the solution for 
the true budget $b$ and then gradually increasing the budget to $b'$ 
and adapting the solution as necessary.

It turns out that a  modification of AMES guarantees that, independently of which equal shares-outcome $(W,X)$ AMES is initialized with,
it will output the same solution $(W^*,X^*)$ in either case. That is,  
we can execute AMES adaptively, starting with budget $b$ and increasing
the budget until an exhaustive solution is found, yet explain the result
as the output of AMES on the final budget $b'$.

Consider two stable solutions $(W,X)$ and $(W',X')$ for the election 
$E = (N,P,(A(i))_{i\in N}, b, \cost)$, and suppose that $W\neq W'$ or $X\neq X'$.
Let $\boldsymbol{\pi}(X')$ and $\boldsymbol{\pi}(X)$ be the corresponding price-per-voter vectors.
Let $j$ be the first index such that $\pi_j(X)=\pi(p,X)\neq \pi(q,X') =\pi_j(X')$ or $\pi_j(X) = \pi(p,X)= \pi(q,X')=\pi_j(X')$ but $p\neq q$.
We will now show that the former case cannot occur.
\begin{restatable}{proposition}{propdiffdueties}\label{prop:diff-due-ties}
	$\pi(p,X)= \pi(q,X')$ and so in particular $p\neq q$.
\end{restatable}
\noindent This shows that the only reason for inconsistency between the outcomes $(W,X)$ and $(W',X')$ can come from ties. 
Specifically, the following scenario may lead to inconsistency:
Suppose we have two voters $1$ and $2$, projects $p_1,p_2,p_3$ where $v_1$ approves $p_1$ and $p_2$ and $v_2$ approves $p_2$ and $p_3$.
The tie-breaking rule for projects is $p_1<p_2<p_3$ and their costs are $cost(p_1)=1$, $cost(p_2)=4$, $cost(p_3)=2$ with a total budget of  $b=5$.
AMES selects outcome $\{p_1,p_2\}$ which are respectively fully paid by their unique supporters. If we increase the budget to $6$ and then initialize AMES on $\{p_1,p_2\}$ with the corresponding load distribution, then $p_2$ would not be selected because the capacity of $v_2$ is strictly less than $2$. The outcome remains $\{p_1,p_2\}$. If, however, we run AMES from scratch with budget $b=6$, the selected outcome is $\{p_1,p_2\}$ with $p_1$ fully paid by voter $1$ and $p_2$ paid equally by both voters.

\smallskip
A simple resolution to this inconsistency is to allow greedy update steps that add a project and (if necessary) remove projects with higher per-voter price \textit{or} 
same per-voter price and lower tie-breaking priority.
Specifically, we introduce a more refined update step,
where the capacities of voters depend on the project under consideration. 
Given a tie-breaking order $>$ on $P$,
for each voter $i\in N$, $p\in P$, and solution $(W,X)$, 
we define $i$'s {\em project-dependent capacity} 
$$
\kappa_{i,p}(X)=\max\{\kappa_i, \max\{x_{i,p'}: p>p',p'\in W\}\}.
$$
We will say that $(W,X)$ is \textit{lexicographically stable} 
if there is no project $p\in P$ and a positive integer $t>N_p(X)$ such that 
$|\{i\in N: p\in A_i, \kappa_{i,p}(X)\ge \frac{cost(p)}{t}\}|\ge t$.
The {\em tie-consistent AMES} proceeds in the same way as AMES, 
but uses lexicographic stability instead of stability. 

Project-dependent capacities may raise the concern that capacities now cannot be computed and sorted for all projects in one sweep; that instead the $O(n\log n)$ sorting cost is incurred for each project, leading to an increase in run-time to $O(m^2n\log n)$. In the proof of the following theorem we provide an implementation of tie-consistent AMES that has run-time $O(mn\log n+m^2n)$.
 
\begin{restatable}{theorem}{amesconsistency}\label{thm:ames-consistency}
Let $(W',X')$ be the outcome of AMES with budget $b'<b$. Let $(W'',X'')$ be the outcome of tie-consistent AMES initialized on $(W',X')$ with budget $b$. Let $(W,X)$ be the outcome of AMES with budget $b$. 
Then $W=W''$ and $X=X''$. Furthermore, each update step of tie-consistent AMES can be completed in time $O(n\log n+nm)$. 
\end{restatable}
\subsection{Skipping Budgets}\label{sec:skippingbudgets}
Another potential for gain in efficiency of MES comes directly from our notion of stability. Fix an election $E = (N,P,(A(i))_{i\in N}, b, \cost)$, and write 
$E(b')$ to denote the election obtained from $E$ by changing the budget to $b'$;
thus, $E=E(b)$.
Consider a stable solution $(W,X)$ for $E(b)$.
It may be the case that $(W,X)$ remains stable for $E(b')$, where $b'>b$.
Now, consider a budget $b''$ with $b<b''<b'$. The following monotonicity 
property is easy to verify.

\begin{proposition}
If $(W,X)$ is unstable for $E(b'')$, then it is also unstable for $E(b')$. 
\end{proposition}
It follows that if $(W,X)$ is stable for $E(b')$, then we do not have to check intermediate budgets $b''$ with $b<b''<b'$. It is then natural to ask 
(1) what is the minimum value $b'>b$ such that $(W,X)$ is unstable for $E(b')$, and 
(2) can we compute this value efficiently?
It turns out that for AMES $b'$ is easy to compute; we remark that, in contrast, for MES
this is not known to be the case (this is because the solutions output by MES are not necessarily equal-share solutions, which makes MES rather fragile with respect to budget modifications).
\begin{restatable}{theorem}{thmunstable}\label{thm:unstable}
    The minimum $b'>b$ such that $(W,X)$ is unstable for $E(b')$ can be computed in time $O(mn^2\log n)$.
\end{restatable}

\section{Experimental Evaluation}\label{sec:exp}
We evaluate the potential gain in efficiency of AMES for the completion problem on real-world participatory budgeting data from three Polish cities. We consider PB data from Pabulib~\cite{stolicki2020pabulib}, an open participatory budgeting library. We select several data sets with a large\footnote{That is, in the range of $40$--$150$, as opposed to $5$--$10$, as in the PB elections in, e.g., Zabrze.} number of proposed projects and mean ballot size of more than one from three Polish cities: Warsaw, Wroclaw and Lodz. Warsaw holds district-based elections, so we consider three districts of Warsaw in our experiments, while the data from Wroclaw and Lodz comes from city-wide elections. The following table summarizes the properties of these data sets.

\begin{figure}[h!]
\begin{tabular}{ |p{3.8cm}||p{0.8cm}|p{1.5cm}|p{1.3cm}|p{1.5cm}|p{2.5cm}|  }
 \hline
 City/District&year &vote count &projects& budget & mean vote length\\
 \hline
 Wroclaw/city-wide & 2018      &53,801&   39 & 4,000,000 &1.87643\\
 Lodz/city-wide & 2020   & 51,472   &151& 5,715,627 & 3.82408\\
 Warszawa/Ursynow &2023 & 6260&  54& 	6,067,849&11.586\\
 Warszawa/Bielany&  2023  & 4,956&  98& 5,258,802&11.3999\\
 Warszawa/Praga-Polodnie &  2023 &   8,922 & 81& 7,180,288 &11.5092\\

 \hline
\end{tabular}
\caption{Participatory budgeting data from Poland}
\end{figure}

\begin{figure*}[h!]
\centering
\begin{subfigure}{0.49\textwidth}
  \includegraphics[width=\linewidth]{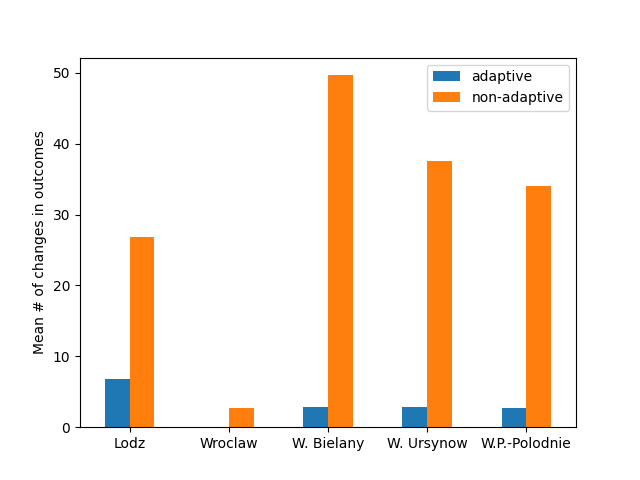}
  \label{fig:comps1}
\end{subfigure}
\hfill
\begin{subfigure}{0.49\textwidth}
  \includegraphics[width=\linewidth]{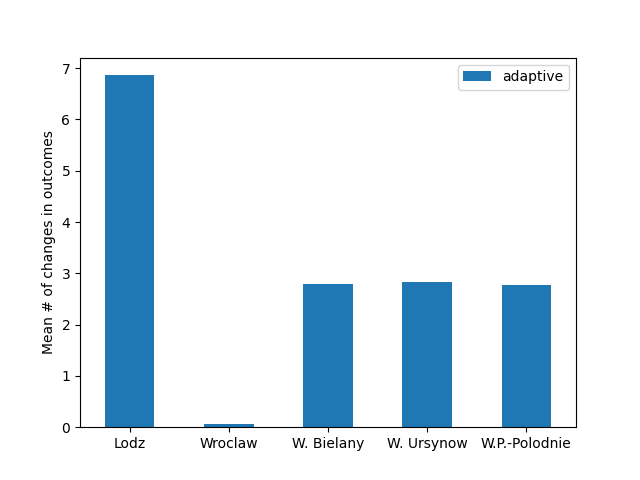}
  \label{fig:comps2}
\end{subfigure}

\caption{We repeatedly top up each voter's budget up by 1 and evaluate the average change in the outcome over $50$ runs. }\label{fig:data}
\end{figure*}
We compare the average difference in outcomes when the budget is increased (i.e., each voter gets more money, as is done in the top-up solution for the completion problem for MES). More formally, for each $i\in [50]$, let
$(W^i, X^i)$ denote the output of AMES with budget of $\frac{b}{n}+i$
per voter. 
We compare solutions at step $i-1$ and step $i$ by calculating the change in outcome as measured by 
$$
|\{p\in P: \pi(p,X^i)\neq \pi(p,X^{i-1})\}|
$$ 
(that is, how many projects have lower per-voter price plus how many were removed), and we average our results over these $50$ runs to obtain a final figure for each data set for the adaptive method. If we run AMES from scratch for each (voter) budget $\frac{b}{n}+ i$, then the number of projects added will be exactly $|W^i|$. Therefore, we take the average number of selected projects over the $50$ runs as the baseline comparison when using AMES non-adaptively.
The results of these experiments are shown in Figure~\ref{fig:data}. Consider Warszawa Bielany; in every run we add $50$ projects on average, while the average difference between consecutive outcomes (and hence the amount of work required) is just below $3$. Similarly, in Lodz, every budget-increasing iteration on average adds over $25$ projects, while the average differences amount to less than $7$.  These results confirm that consecutive outcomes are very similar on average, and hence our budget-adaptive approach offers significant savings over repeatedly executing MES from scratch repeatedly, as measured by the number of iterations that modify the outcome.

\section{Conclusion}
We introduced a new voting rule for PB with approval utilities. 
This rule, which we call Adaptive Method of Equal Shares (AMES), 
can leverage similarities between an equal-shares initialization $(W,X)$ for budget $b'$ and output solution $(W^*,X^*)$ for budget $b''\geq b'$ when run from scratch, allowing it to find $(W^*,X^*)$ in fewer iterations. This feature is relevant for the completion problem of MES-style methods, as it alleviates the need to run the method repeatedly from scratch for different budgets. The solutions output by AMES satisfy the proportionality axiom EJR, and this can be verified efficiently in time $O(n\log n+mn)$.

Anecdotal evidence suggests that in practice there is a dislike of voting rules that are difficult to explain. One may worry that AMES suffers from this drawback to a greater extent than its non-adaptive cousin MES. However, we emphasize that it suffices to explain the simpler rule, which runs from scratch with a fixed budget $b'\ge b$: the adaptive version is only used for finding $b'$ efficiently, by leveraging the similarity in consecutive outcomes. Thus, AMES is essentially as transparent as the currently used completion method for MES: while the choice of the virtual budget $b'$ may appear ``magical'', it can be justified by showing that a further virtual budget increase would result in a proposal that is not feasible with respect to $b$.}

To evaluate the potential of our budget-adaptive approach, we analyzed real-world PB data, and showed that solutions of AMES obtained by repeatedly topping up budgets 
differ very little on average. More broadly, AMES can also be useful when the input election is changed in other ways, e.g., the costs of some projects are updated, or a few voters change their preferences; in such cases, too, checking for updates may be faster than re-computing solutions from scratch.

\smallskip
Our experiments were conducted in Python. To analyze the gain in efficiency by using AMES budget-adaptively in practice, it would be useful to implement fully optimized versions of both adaptive and non-adaptive completion methods in C. This will allow us to compare the constant overhead of an adaptive method with its gain in efficiency from leveraging similarities. Improving upon the run-time of $O(mn^2\log n)$ to skip budgets (Theorem \ref{thm:unstable}) by removing the quadratic dependency on $n$ is an important direction for future work.
\bibliographystyle{abbrv}
\bibliography{bibliography}


\appendix

\newpage
\section{Appendix A}

\subsection{Omitted proofs from Section~\ref{sec:stab+ejr}}

\lemrestrict*
\begin{proof}
Fix an election $E = (N,P,(A(i))_{i\in N}, b, \cost)$, a set of voters $V\subseteq N$, a set of projects $T\subseteq P$, an election $E'$, solutions $(W, X)\in\calWX^*(E)$
and $(W', X')\in\calWX^*(E')$, and 
a voter $i\in V$ as in the statement of the lemma, and let $r = |\{p: x'_{i, p}>0\}|$.
Note that $|A(i)\cap W'|\ge r$.
We will show that there exists a voter $j\in V$ such that voters $i$ and $j$ jointly approve at least $r$ projects in $W$.

To this end, consider the restriction of the per-voter price vector $\boldsymbol{\pi}(X')$  to the set of projects $\{p: x'_{i, p}>0\} =\{p_1,p_2,\ldots, p_r\}$, i.e., the vector
\begin{align}
\boldsymbol{\pi}'=\left(\frac{cost(p_1)}{|V_1|},\ldots,\frac{cost(p_r)}{|V_r|}\right),
\end{align}
where  for each $q\in [r]$ the set $V_q\subseteq V$ consists of voters in $V$ 
who pay for $p_q$ in $(W', X')$,  
and we assume that 
$\frac{cost(p_1)}{|V_1|}\leq \frac{cost(p_2)}{|V_2|}\leq \ldots \leq \frac{cost(p_r)}{|V_r|}$.
Let $\ell$ be the index of the first project $p_{\ell}$ in $\boldsymbol{\pi}'$ such that some voter in $V_{\ell}$ does not pay for $p_\ell$ in $(W,X)$.
Let $V'$ be the set of voters paying for $p_{\ell}$ in $(W,X)$; 
note that $V'$ may be empty if $p_{\ell}\notin W$. 

Observe that $V_{\ell}\setminus V'\neq \varnothing$ and so 
 $|V_{\ell}\cup V'| > |V'|$; also, we have $|V_{\ell}\cup V'|\ge |V_\ell|$. 
 Hence, if every $j\in V_{\ell}\setminus V'$ had capacity $\kappa_j(X)\geq  \frac{\cost(p_\ell)}{|V_{\ell}|}$ in $(W, X)$, 
 the voters in $V_\ell\cup V'$
 could share the cost of $p_\ell$ by paying 
 $\frac{\cost(p_\ell)}{|V_\ell \cup V'|}$ each;
 as 
 $\frac{\cost(p_\ell)}{|V_\ell \cup V'|}< \frac{\cost(p_\ell)}{|V'|}$, 
 this would contradict the stability of $(W,X)$.

 It follows that there exists a voter $j\in V_{\ell}$ with capacity $\kappa_j(X)<\frac{\cost(p_{\ell})}{|V_{\ell}|}$ who approves $p_{\ell}$, 
 but does not pay for $p_{\ell}$ in $(W,X)$.
	This has two important implications:
\begin{itemize}
\item[(a)] Voter $j$ contributes strictly more than $\frac{b}{n}-\frac{\cost(p_{\ell})}{|V_{\ell}|}$ towards $W$.
\item[(b)] Voter $j$ contributes at most $\frac{\cost(p_{\ell})}{|V_{\ell}|}$ towards each project in $W$.
\end{itemize}

Let $P' =\{p_1, \dots, p_{\ell-1}\}$, and 
note that by our choice of $\ell$ we have $P'\subseteq W$, 
and also $P' \subseteq A(i)$.
Suppose that in $(W, X)$ voter $j$ contributes to $y$ projects in $W\setminus P'$. 
To complete the proof, it suffices to show that $y\ge r-|P'|= r-\ell+1$.

Consider voter $i$. Her budget is $\frac{b}{n}$, and in $(W', X')$
she contributes $\frac{\cost(p_q)}{|V_q|}$ to project 
$p_q$, for $q=1, \dots, r$. Therefore, 
\begin{equation}\label{eq:voter-i}
    \frac{b}{n}-\sum_{q=1}^{\ell-1}\frac{\cost(p_q)}{|V_q|}\geq\sum_{q=\ell}^{r}\frac{\cost(p_q)}{|V_q|}.
\end{equation}
Now, consider voter $j$ and solution $(W, X)$.
By our choice of $\ell$ the cost of each project $p_q\in P'$
is shared by at least $|V_q|$ voters in $(W, X)$. Hence, in $(W, X)$ voter $j$
spends at most $\sum_{q=1}^{\ell-1}\frac{\cost(p_q)}{|V_q|}$ on projects in
$P'$. Therefore, by observation~(a) she spends strictly more than
$$
\frac{b}{n} - \frac{\cost(p_{\ell})}{|V_{\ell}|} - \sum_{q=1}^{\ell-1}\frac{\cost(p_q)}{|V_q|}
$$
on projects in $W\setminus P'$; by~\eqref{eq:voter-i}, this quantity is at least 
$$
\sum_{q=\ell+1}^{r}\frac{\cost(p_q)}{|V_q|}.
$$
On the other hand, in $(W, X)$, voter $j$ contributes to $y$ projects in $W\setminus P'$, 
and, by observation~(b), she spends at most $\frac{\cost(p_{\ell})}{|V_{\ell}|}$
on each of them. Hence, we obtain
$$
y\cdot \frac{\cost(p_{\ell})}{|V_{\ell}|}> 
\sum_{q=\ell+1}^{r}\frac{\cost(p_q)}{|V_q|}
$$
and hence
$$
y> \sum_{q=\ell+1}^{r}\frac{\cost(p_q)}{|V_q|} \times\frac{|V_{\ell}|}{cost(p_{\ell})}
\ge \sum_{q=\ell+1}^{r} 1=r-\ell,
$$
where we use the fact that $\frac{\cost(p_q)}{|V_q|}\geq \frac{\cost(p_{\ell})}{|V_{\ell}|}$ for $q\geq \ell$.
Since $y$ is an integer, we have $y\ge r-\ell+1$, which is what we set out to prove.
Hence, the proof is complete.
\qed\end{proof}

\propverifycomplexity*
\begin{proof}

Let $(W, X)$ be an equal-shares solution. The capacity $\kappa_i$ of each $i\in N$ can be computed in time $O(|W|)=O(m)$. We can then sort the voters by their capacity (in non-increasing order) in time $O(n\log n)$. Then, for each $p\in P$, let $\alpha(p)$ be the maximum value of $t$ 
such that the capacity of each of the first $t$ voters in the list who approve $p$ is at least $\frac{\cost(p)}{t}$; we set $\alpha(p)=-\infty$ if 
there is no positive value of $t$ with this property. The quantity $\frac{\cost(p)}{\alpha(p)}$, which can be computed in time $O(n)$ by scanning through the list, is the minimum per-voter price that can be attained for $p$
(either by adding it to the outcome or by increasing the number of voters who contribute towards it).
Observe that the solution $(W,X)$ is stable if and only if 
for each $p\in W$ we have $\frac{\cost(p)}{\alpha(p)}\geq \pi(p, X)$
and for each $p\not\in W$ we have $\alpha(p)=-\infty$.
Hence, this process is indeed a verification process for stability, and, by the above analysis, has a running time of $O(n\log n+nm)$.
\qed\end{proof}
\subsection{Omitted proofs from Section~\ref{sec:AMES}}
\propupdate*
\begin{proof}
Suppose that $W'$ is obtained from $W$
by lowering the per-voter price of project $p$ and removing a set of projects $W^-$, so that in $(W', X')$ the load of $p$ is shared by a set of voters $V$ of size $|V|=s$.
By construction, $X'$ is an equal-shares load distribution for $W'$. To see that it is feasible, 
consider a voter $i\in N$. There are four cases.
\begin{itemize}
\item[{\bf Case 1: }] $i\not\in V$. 

Then $x'_{i, p}=0$ and hence
$X'_i\le X_i\le \frac{b}{n}$
(as $(W, X)$ is a feasible solution). 

\item[{\bf Case 2: }]
$i\in V$ and 
$\frac{\cost(p)}{s} \le \frac{b}{n} - X_i$.

Then
$X'_i\le \frac{\cost(p)}{s} + X_i\le \frac{b}{n}$, so this voter stays within the budget as well.

\item[{\bf Case 3: }]
$i\in V$, $\frac{\cost(p)}{s} > \frac{b}{n} - X_i$, 
and $x_{i, p}=0$. 

Note that $i\in V$ implies $\frac{\cost(p)}{s}\le \kappa_i$. 
Moreover, since $\frac{\cost(p)}{s} > \frac{b}{n} - X_i$ , the set $W^-$ contains a project
$p'$ with 
$x_{i, p'}>0$,
$\pi(p', X)= \max\{\pi(p, X): x_{i, p}>0, p\in W\}$,
such that $\kappa_i= x_{i, p'}-\epsilon$
(recall that we define $\epsilon=\min_{x,y\in K, x\neq y}|x-y|$, 
where $K=\{\frac{cost(p)}{i}: p\in P, i\in N\}$).
Hence, 
$$
x_{i, p'} = \kappa_i+\epsilon \ge \frac{\cost(p)}{s}+\epsilon>\frac{\cost(p)}{s}.
$$
Therefore, $X'_i\le \frac{\cost(p)}{s} + \sum_{q\in W}x_{i, q} - x_{i, p'}< X_i\le \frac{b}{n}$:
the increase in the total load of $i$ caused
by making $i$ contribute towards $p$ is counteracted by a decrease in her total load from removing $p'$. 
\item[{\bf Case 4: }]
$i\in V$, $\frac{\cost(p)}{s} > \frac{b}{n} - X_i$, 
and $x_{i, p}>0$. 

Note that $x_{i, p}>0$ implies that $p\in W$.
Since the per-voter price of $p$ goes down when moving from $W$ to $W'$,
we have $x_{i,p}>x'_{i, p} = \frac{cost(p)}{s}$.
Hence, 
$$
X'_i\le \frac{\cost(p)}{s} + \sum_{q\in W}x_{i, q} - x_{i, p}< X_i\le \frac{b}{n}.
$$
\end{itemize}

We conclude that $(W', X')\in \calWX^=(E)$.

It remains to argue that $(W', X')\rhd (W, X)$.
To show this, 
we will prove that $\pi(p', X)> \pi(p, X')$ for each $p'\in W^-$, i.e., the per-voter price of each project in $W^-$ with respect to $X$ is strictly higher than the per-voter price of $p$ with respect to $X'$.
Indeed, fix $p'\in W^-$. Project $p'$ was placed in $W^-$ because 
$\kappa_i= x_{i,p'}-\epsilon$ for some $i\in V$, 
and therefore
$\pi(p', X) =  x_{i, p'}=\kappa_i+\epsilon > \kappa_i$. On the other hand, 
$\pi(p, X')=\frac{\cost(p)}{s}\le \kappa_i$ since $i\in V$.
We conclude that $\pi(p, X') < \pi(p', X)$. As this holds for every $p'\in W^-$ and, moreover, $\pi(p,X')<\pi(p,X)$, we have $(W', X')\rhd (W, X)$.
\qed\end{proof}
\lemdecappeal*
\begin{proof}
Let $V'$ and $V''$ be the sets of voters who pay 
for $p'$ and $p''$ in $X'$ and $X''$, respectively:
$$
V'=\{i\in N: x'_{i, p'} > 0\}, \qquad
V''=\{i\in N: x''_{i, p''} > 0\}.
$$
Suppose for the sake of contradiction that 
$\pi(p', X') = \frac{\cost(p')}{|V'|} > \frac{\cost(p'')}{|V''|} = \pi(p'', X'')$.
We will argue that 
$\kappa_i(X)\ge \frac{\cost(p'')}{|V''|}$ for each 
$i\in V''$; as all voters in $V''$ approve $p''$, 
this means that, given $(W, X)$, 
we could have added $p''$ with per-voter price $\frac{\cost(p'')}{|V''|}$ instead of adding $p'$ with per-voter price $\frac{\cost(p')}{|V'|}$, a contradiction with
$(W, X)\xrightarrow{p'} (W', X')$ being a greedy update step.

Indeed, consider a voter $i\in V''$. We have $\kappa_i(X')\ge \frac{\cost(p'')}{|V''|}$.
Now, if $\kappa_i(X)\ge \kappa_i(X')$, we are done. Otherwise, 
the capacity $\kappa_i$ of voter $i$ increased as we transitioned from
$(W, X)$ to $(W', X')$. 
This can happen in one of two ways:
\begin{itemize} 
\item[(1)] $\max_{p\in W} x_{i,p}<\max_{p\in W'} x'_{i,p}$ or 
\item[(2)] $\frac{b}{n}-X_i<\frac{b}{n}- X'_i$.
\end{itemize}
Case (1) can only happen if $i$ pays for project $p'$ in $X'$. But this implies that $i\in V'$,
so $\kappa_i(X)\geq \frac{\cost(p')}{|V'|}>\frac{\cost(p'')}{|V''|}$.
Case (2) can only happen
if $x_{i, p}>0$ for some project $p$ that was removed
from $W$ in order to accommodate $p'$. As we only remove projects with higher per-voter price, we have
$$
x_{i, p} = \pi(p, X)> \pi(p', X')=\frac{\cost(p')}{|V'|}
$$
and hence $x_{i,p}>\frac{\cost(p')}{|V'|}
\geq \frac{\cost(p'')}{|V''|}+\epsilon$, 
where the second inequality follows by our choice of $\epsilon$.
As $\kappa_i(X)\ge x_{i,p}-\epsilon$, we conclude that 
$\kappa_i(X)
>\frac{\cost(p'')}{|V''|}$ in this case too.    

In each case, we have $\kappa_i(X)\ge \frac{\cost(p'')}{|V''|}$ for all $i\in V''$, 
which means that lowering the per-voter price of $p''$ to 
$\frac{cost(p'')}{|V''|}$ was a feasible update step when processing $(W, X)$. Together with
$\frac{\cost(p')}{|V'|} >  \frac{\cost(p'')}{|V''|}$, this is a 
contradiction to $(W, X)\xrightarrow{p'} (W', X')$ being a greedy update step. 
\qed\end{proof}

\propgreedy*
\begin{proof}
We will prove that if a greedy update step  lowers the per-voter price of a project $p$ (by adding a project $p$ to the outcome or increasing the number of voters paying for it) then $p$'s per-voter price is not changed in subsequent steps. 
This implies that if $(W^*,X^*)$ is the result of $t$ greedy update steps starting from solution $(W,X)$, then  $t=|\{p\in P: \pi(p,X^*)<\pi(p,X)\}|$, proving our claim.

Indeed, consider a project $p$
whose per-voter price goes down during a greedy update step 
$(W,X)\rightarrow (W^1,X^1)$, and let $V$ be the voters who contribute to $p$ in $X^1$. Suppose for the sake of contradiction that $p$'s per-voter price changes during a subsequent update step. 
We consider the first such update step $(W^2,X^2)\rightarrow (W^3,X^3)$
(note that it may happen that $(W^1, X^1)=(W^2, X^2)$).
\begin{enumerate}
\item[\textbf{Case 1}]
$(W^2,X^2)\xrightarrow{p} (W^3,X^3)$. 
In this case, $p\in W^3$, so $p$ is not removed during this update step. This means that $p$'s per-voter price does not increase, as removal is the only means to increase a project's per-voter price. However, applying Lemma~\ref{lem:dec-appeal} repeatedly,
we conclude that $\pi(p,X^3)\geq \pi(p,X^1)$, so $p$'s per-voter price did not increase either.
\item[\textbf{Case 2}] 
$(W^2,X^2)\xrightarrow[p]{p'} (W^3,X^3)$. 
Suppose project $p$ was removed while another project $p'$ was added to the outcome, 
so $W^3\setminus W^2=\{p'\}$. By Lemma~\ref{lem:dec-appeal} the per-voter price of $p'$ in the resulting load distribution $X^3$ is at least as high as the per-voter price of $p$ at the time when it was originally added. By our choice of update step, $\pi(p,X^2)=\pi(p,X^1)$, 
so the per-voter price of $p$ has not changed since the update step $(W,X)\rightarrow (W^1,X^1)$. This means that, when adding $p'$, we removed a project whose per-voter price was no higher than that of $p'$, a contradiction with how we select the set $W^-$ during an update step.
\end{enumerate}  
\qed\end{proof}

\proplsoptmesruntime*
\begin{proof}
The algorithm proceeds by starting with an arbitrary solution and performing greedy update steps; by Proposition~\ref{prop:greedy-k}, it converges after 
$|\{p\in P: \pi(p,X^*)< \pi(p,X)\}|=O(|W^*|)=O(m)$ steps.

It remains to evaluate the time complexity of each step. 
Let $(W, X)$ be the solution at the start of a step.
For each voter $i\in N$ we can compute $r_i\in\argmax_{i\in W}\{x_{i,r}\}$ and capacity $\kappa_i$ in time $O(|W|)=O(m)$. We can then sort the voters by their capacity (in non-increasing order) in time $O(n\log n)$. As in the proof of 
Theorem~\ref{prop:verifycomplexity}, for each $p\in P$, let $\alpha(p)$ be the maximum $t$ such that the capacity of each of the first $t$ voters in the list who approve $p$ is at least $\frac{\cost(p)}{t}$; we set $\alpha(p)=-\infty$ 
if there is no positive value of $t$ with this property. 
Scanning though the list $(\alpha(p))_{p\in P}$, we terminate if $\alpha(p)=-\infty$ for all $p \in P\setminus W$ and $\frac{\cost(p)}{\alpha(p)}\geq \pi(p, X)$ for each $p\in W$. Otherwise,
we select a project $p\in P$ with the minimum value of $\frac{cost(p)}{\alpha(p)}$. Clearly, this can be done in time $O(m)$.
To accommodate project $p$, we may need to remove some set of projects $W^{-}\subset W$, which we compute as follows.
We set $W^{-}=\varnothing$. For each new contributor $i$ to $p$ (i.e. $x_{i,p}=0$), if $\frac{b}{n}-X_i<\frac{cost(p)}{\alpha(p)}$, we add $r_i$ to $W^{-}$: $W^{-}=W^{-}\cup \{r_i\}$.
We set $W=W\cup \{p\}\setminus W^{-}$ and update the load distribution $X$.
The running time of the selection process is thus $O(n\log n+nm)$, whereas the load distribution $X$ can then be updated in time $O(nm)$.
\qed\end{proof}


\subsection{Omitted proofs from Section~\ref{sec:ames-consistency}}
\propdiffdueties*
\begin{proof}
By our choice of the index $j$ there is a list of projects
$q_1, \dots, q_{j-1}$ such that for all $\ell<j$ we have
$\pi_\ell(X)=\pi(q_\ell, X)=\pi(q_\ell, X')=\pi_\ell(X')$.
Let $V_\ell$ be the set of voters who pay for $q_\ell$ in $(W,X)$, 
and  let $V_\ell'$ be the set of voters who pay for $q_\ell$ in $(W',X')$. 
We claim that $V_\ell=V_\ell'$ for each $\ell<j$.

Indeed, if this is not the case, let $\ell<j$ be the first index 
for which this claim does not hold, i.e., $V_{\ell}\neq V_{\ell}'$.
Note that $|V_\ell|=|V'_\ell|$; hence,  
there exists a voter $i$ that contributes towards 
$q_\ell$ in $(W,X)$, but not in $(W',X')$, i.e., $i\in V_\ell\setminus V'_\ell$. 
We will argue that that cost of $q_\ell$ can be shared equally by voters in 
$V'_\ell\cup\{i\}$, thereby contradicting the stability of $(W',X')$.

To see this, let $z = |V_\ell|=|V'_\ell|$, and
note that, by the choice of $\ell$, for each $\ell'<\ell$ voter $i$'s contribution to project $q_{\ell'}$ is the same in $(W,X)$ and in $(W',X')$.
This means, in particular, that upon $o$ contributing to $q_1,\ldots, q_{\ell-1}$, $i$ still has $\frac{\cost(q_{\ell})}{z}$ budget left over. Furthermore, the per-voter price of all other projects it pays for in $(W',X')$ is at least as high as that of $q_{\ell}$, and hence she pays at least $\frac{\cost(q_{\ell})}{z}$ for each. This implies that $\kappa_i(X')\geq \frac{\cost(q_{\ell})}{z}-\epsilon\geq
\frac{\cost(q_{\ell})}{z+1}$.
On the other hand, in $(W', X')$ every $i'\in V'_{\ell}$ contributes 
$\frac{\cost(q_{\ell})}{z}$ 
to project $q_{\ell}$, so
her capacity is $\kappa_{i'}(X')\geq \frac{\cost(q_{\ell})}{z}-\epsilon\geq \frac{\cost(q_{\ell})}{z+1}$. This implies that $(W', X')$ is unstable, 
as we can increase the number of voters who share the cost of $q_\ell$ to $z+1$, 
a contradiction. Thus, we conclude that $V_\ell=V_\ell'$ for all $\ell<j$.

Now	suppose $\pi_j\neq \pi'_j$; without loss of generality assume 
$\pi_j < \pi'_j$. We consider two cases.
	\begin{itemize}
	\item[\textbf{Case 1}] $p=q$: Since $\pi_j<\pi'_j$ and $p=q$, we have $|V_j|>|V'_j|$, i.e., 
	more voters pay for $p$ in $(W,X)$ than in $(W',X')$. 
	Thus, there exists a voter
	 $i\in V_j\setminus V'_j$. By the same argument as above, it follows that $\kappa_{i'}(X')\geq \frac{\cost(p)}{|V'_j|+1}$ for $i'\in V'_j\cup\{i\}$, thereby contradicting the stability of $(W',X')$.
	\item[\textbf{Case 2}] $p\neq q$: Since $\pi_j < \pi'_j$, this implies in particular that $\pi_j>0$ and $V_j\neq \varnothing$. 
	Each $i\in V_j$ contributes the same amount to projects $q_\ell$, $\ell <j$, in both $(W,X)$ and $(W', X')$. Each project $q_\ell, \ell\geq j$,
	that $i$ contributes to in $(W',X')$ has per-voter price of at least 
	$\pi'_j>\pi_j=\frac{\cost(p_j)}{|V_j|}$. Therefore after paying for projects $p'_\ell$, $\ell<j$, each such $i$ has at least $\frac{\cost(p_\ell)}{|V_\ell|}$ out of the $\frac{b}{n}$
	budget remaining. Moreover, for every $\ell>j$ her contribution to 
    $p'_\ell$ in $(W', X')$
	is at least $\frac{\cost(p'_j)}{|V'_j|}=
	\pi'_j < \pi_j=\frac{\cost(p_j)}{|V_j|}$. 
	Hence, for each $i\in V_j$ we have 
	$\kappa_i(X')\geq \max\{\frac{\cost(p_j)}{|V_j|},\frac{\cost(p'_j)}{|V'_j|}-\epsilon\}\geq \frac{\cost(p_j)}{|V_j|}$. This means that, given solution $(W', X')$, voters in $V_j$ can afford $p_j$.	This contradicts the stability of $(W',X')$.
	\end{itemize}
Hence, $\pi_j=\pi'_j$ and hence $p\neq q$, which is what we wanted to prove.
\qed\end{proof}

\amesconsistency*
\begin{proof}
First observe that by design, tie-consistent AMES outputs lexicographically stable solutions. 
Also observe that AMES run from scratch outputs a lexicographic solution, since it respects the tie-breaking rule.
So we will show that if $(W,X)$ and $(W',X')$ are two lexicographically stable solutions, then $W=W'$ and $X=X'$. 

As in Section~\ref{sec:ames-consistency}, suppose that $W\neq W'$ or $X\neq X'$.
Let $\boldsymbol{\pi}(X')$ and $\boldsymbol{\pi}(X)$ be the corresponding per-voter price vectors.
Let $j$ be the first index such that 
$\pi_j(X)=\pi(p_j,X)\neq \pi(p'_j,X') =\pi_j(X')$ or 
$\pi_j(X)=\pi(p_j,X) =  \pi(p'_j,X') =\pi_j(X')$, 
but $p_j\neq p'_j$.
Since lexicographically stable solutions are stable solutions, 
Proposition~\ref{prop:diff-due-ties} applies, showing that the per-voter prices of $p_j$ and $p'_j$ are equal: $\pi(p_j,X)= \pi(p'_j,X')$ and therefore $p_j\neq p_j'$.
Without loss of generality suppose $p_j>_{\text{lex}} p'_j$. Let $V$ be the set of voters who pay for $p_j$ in $(W,X)$. We claim that for every $i\in V$ we have $\kappa_{i,p_j}(X')\geq \frac{cost(p_j)}{|V|}$, contradicting the lexicographic stability of $(W',X')$. Since $i$ pays for $p_j$ in $(W,X)$ after having contributed 
the same amount in $(W',X')$ and $(W,X)$
to lower-price projects $p_1=p'_1,\ldots, p_{j-1}=p'_{j-1}$, 
she still has at least $\frac{cost(p_j)}{|V|}$ budget available. However, all projects $p'_\ell$, $\ell>j$, have $\pi(p'_\ell,X')> \pi(p'_j,X')$ or $\pi(p'_\ell,X')=\pi(p'_j,X')$ and $p'_j>_{\text{lex}} p'_\ell$.
So indeed, if $i$ does pay for any $p'_\ell$ with $\ell>j$, then 
$$
\kappa_{i,p_j}(X')\geq \pi(p'_j,X')\geq \pi(p'_j,X')=\frac{cost(p'_j)}{|V'|}=\frac{cost(p_j)}{|V|},
$$
and if not, then $i$'s available budget is at least $\frac{cost(p_j)}{|V|}$, so $\kappa_{i,p_j}(X')\geq\frac{cost(p_j)}{|V|}$. This gives us our desired contradiction. We conclude that such a $j$ does not exist,
and $\boldsymbol{\pi}(X')=\boldsymbol{\pi}(X)$, implying that $W=W'$ and $X=X'$, as desired.

Next we discuss how to calculate the project-dependent capacities in tie-consistent AMES so that an update step has run-time $O(n\log n+mn)$. Consider a solution $(W, X)$.
Just like in AMES, we calculate the capacities $\kappa_i(X)$ for all voters $i\in N$, and for each voter, we keep track of a project in $(W,X)$ with the highest per-voter price that she pays for (breaking ties lexicographically).
We sort $\{\kappa_i(X)\}_{i\in N}$ in time $O(n\log n)$. 

When considering a project $p\in P$ to add, for every supporter $i$ of $p$ we increase $i$'s capacity by $\epsilon$ if $\kappa_i(X)>\frac{b}{n}-X_i$ and the highest per-voter price (lowest priority) project $p'$ that $i$ contributes to has lower priority than $p$.
That is $\kappa_i(X)=X_{i,p'}-\epsilon$ and we set $\kappa_{i,p}(X)=X_{i,p'}$.
This may result in the list $\kappa_{i,p}(X)$ not being sorted in increasing order. However, note that each entry in the resulting capacity array was either left unchanged or increased by $\epsilon$. That is, the capacity array can be viewed as an interleaving
of two sorted subarrays: one for the unchanged entries and one for the updated entries. Hence, to obtain the sorted project-dependent capacities, we execute a simple merging algorithm that runs in time $O(n)$.
\qed\end{proof}

\subsection{Omitted proofs from Section \ref{sec:skippingbudgets}}

\thmunstable*
\begin{proof}
For each project $p\in P$ we will compute the minimum increase in budget needed for $p$ to get added to the output or, if it is already contained in the output, 
to increase the number of voters who pay for it. We then output the minimum of these values over all projects in $P$.

Fix a project $p$. Note that in an equal shares outcome a voter $i\in N$ approving project $p$ can only contribute values from the finite set 
$$
S=\left\{0,\cost(p),\frac{\cost(p)}{2},\ldots,\frac{\cost(p)}{|\{j\in N: p\in A(j)\}|}\right\}
$$
towards $p$; observe that $|S|\leq n+1$.
For every positive integer $t\le |\{j\in N: p\in A(j)\}|$ 
and for each voter $i$ who approves $p$, we compute the amount $y_{i, t}$ by which we would need to increase the budget of $i$ so that her capacity exceeds $\frac{\cost(p)}{t}$; this computation can be performed in time $O(n)$.
Sorting the list $(y_{i, t})_{i\in N}$ in non-decreasing order in time $O(n\log n)$, we consider the $t$-th smallest entry in the sorted array. This value, which we denote by $z_{p, t}$,
is the smallest amount by which we would need to increase the voters' budgets so that project $p$ could be paid for by $t$ of its supporters. We then pick the best value of $t$ for project $p$, i.e., we let $z_p=\min_t z_{p, t}$. We output $n\cdot \min_{p\in P}z_p$. The entire procedure runs in time $O(mn^2\log n)$.
\qed\end{proof}
\section{Additional Material for Section \ref{sec:stab+ejr}: Weak Stability}\label{sec:add}
We discuss a weaker notion of stability, 
which only considers projects not currently included in the outcome. 
The formal definition is as follows.
\begin{definition}\label{def:ws}
    A solution $(W, X)\in\calWX^=(E)$ is {\em weakly stable} if there is a project 
$p\notin W$ and a $t\in\mathbb N$ such that $|\{i\in N: p\in A(i), \kappa_i(X)\ge \frac{\cost(p)}{t}\}|\ge t$.
\end{definition}
The following example shows that weakly stable solutions may fail to provide EJR 
in the approval PB setting.
\begin{example}\label{ex:stability}
Consider election $E$ from Example~\ref{ex:run}. In $E$ every voter approves every project, and $\sum_{i=1}^5 \cost(p_i)=35=b$. Hence $N$ forms a $P$-cohesive group, 
and therefore an outcome satisfying EJR must contain all projects, as otherwise no voter would have approval utility $5$.
Now consider the outcome $(W,X)$ with 
\begin{align}
&W=\{p_1,p_2,p_3,p_4\},\label{outcomestart}\\ 
&x_{1,p}=x_{2,p}=3 \text{ for } p\in\{p_1, p_2, p_3\}, \ x_{1,p}=x_{2,p}=0
\text{ for }p\in \{p_4,p_5\}\\
&x_{3,p_4}=7\text{ and  }x_{3,p}=0 \text{ for }p\neq p_4.\label{outcomeend}
\end{align} 
Clearly, load distribution $X$ is priceable and equal-shares.
 In total, each of voters $1$ and $2$ spends 
 $$
 \sum_{j=1}^3 x_{1,p_j}=\sum_{j=1}^3 x_{2,p_j}=9< \frac{35}{3},
 $$
 and has remaining budget of $\frac{8}{3}$. This implies that their capacities are $\kappa_1(X)=\kappa_2(X)\leq \max\{\frac{8}{3},3\}=3$.
 Voter $3$ pays $x_{3,p_4}=7\leq \frac{35}{3}$ and her remaining budget is $\frac{14}{3}$; hence, $\kappa_3(X)\leq 7$.
 We claim that no $p\in P\setminus W$ satisfies $|\{i\in N: p\in A(i), \kappa_i(X)\ge \frac{\cost(p)}{t}\}|\ge t$ for a positive integer $t$.
 Indeed,  $P\setminus W = \{p_5\}$. Moreover, splitting $p_5$ equally among the three voters would cost $\frac{10}{3}$ for each of them, 
 and $\kappa_1(X)=\kappa_2(X) = 3 <\frac{10}{3}$. 
 On the other hand, $3$ cannot afford $p_5$ on her own, since $\kappa_3(X)< 10$.
 \end{example}

 A special case of participatory budgeting is multiwinner voting: 
 Here, the goal is to select a committee of size $k\in \mathbb{N}$ given voters' approval ballots. Thus, an instance of multiwinner voting with target committee size $k$
 can be viewed as an instance of participatory budgeting where each project has unit cost
 and the budget is $k$.
 It turns out that, unlike in the general case,
in the setting of multiwinner voting weak stability (Definition~\ref{def:ws}) implies EJR. 

Fix an instance of participatory budgeting 
$E = (N, P, (A(i))_{i\in N}, k, \cost)$, where $\cost(p)=1$ for each $p\in P$,
and an equal-shares solution $(W, X)$ for $E$. 
For each voter $i\in N$, let $z_i=\max\{x_{i, p}: p\in W\}$.
Note that, since $X$ is an equal-shares load distribution, if $z_i>0$, we have $z_i=\pi(p, X)$ for some $p\in W$; more precisely, $p$ has the highest per-voter price among the projects that
voter $i$ contributes to.
Let 
$$
\kappa_i(X)=\max\left\{z_i-\epsilon, \frac{k}{n}-X_i\right\}, 
$$
where $\epsilon=\frac{1}{n(n-1)}$
(recall that for general project costs we define 
$$
\epsilon=\min_{x,y\in K, x\neq y}|x-y|, \qquad\text{where }
K=\left\{\frac{\cost(p)}{i}: p\in P, i\in N\right\};
$$
the reader can verify that 
for unit costs this expression simplifies to $\frac{1}{n(n-1)}$). 

\begin{proposition}\label{prop:committee-ejr}
If for every $p\notin W$ and every $t\in [n]$ it holds that $|\{i\in N: p\in A(i), \kappa_i(X)\geq \frac{1}{t}\}|<t$ then $W$ provides EJR.
\end{proposition}
\begin{proof}
Suppose for the sake of contradiction that $W$ does not provide EJR. Then there exists a positive integer $\ell\le k$
and a group of voters $V$ with $|V|\ge \ell\cdot\frac{n}{k}$ such that $|\cap_{i\in V}A(i)|\ge \ell$, yet $|A(i)\cap W|<\ell$ for all $i\in V$. Note that there exists a project 
$p^*\in (\cap_{i\in V}A(i))\setminus W$.
Since $(W,X)$ is weakly stable, $\kappa_i < \frac{1}{|V|}$ for some $i\in V$.
Thus, voter $i$ spends more than $\frac{k}{n}-\frac{1}{|V|}$ on at most $\ell-1$ members of $W$. By the pigeonhole principle, this means that for some $p\in W$, $i$ pays $x_{i,p}>\frac{\frac{k}{n}-\frac{1}{|V|}}{\ell-1}=\frac{k}{n\ell}\geq \frac{1}{|V|}.$
Since $X$ is an equal-shares distribution, the quantity $\frac{1}{x_{i, p}}$ is a positive integer, and it follows
that $\frac{1}{x_{i, p}}\le |V|-1$ and hence $x_{i, p}\ge \frac{1}{|V|-1}$. 
Further, by construction, we have $\kappa_i\ge x_{i, p}-\frac{1}{n(n-1)}
\ge \frac{1}{|V|-1}-\frac{1}{n(n-1)}$.
As we have 
$\frac{1}{|V|-1}-\frac{1}{|V|}\ge \frac{1}{n(n-1)}$, 
this implies $\kappa_i\ge \frac{1}{|V|}$, a contradiction.
\qed\end{proof}


\end{document}